\documentclass[preprint,12pt]{elsarticle}
\usepackage{graphicx,amssymb,amsmath,subfig,numprint}
\journal{Theoretical Computer Science}
\pdfoutput=1

\newtheorem{theorem}{Theorem}
\newtheorem{lemma}{Lemma}
\newtheorem{definition}{Definition}

\newenvironment{proof}{\noindent \emph{Proof. }}{\hfill \hbox{\rlap{$\sqcap$}$\sqcup$}\\}

\begin{document}

\begin{frontmatter}

\title{Distances on Rhombus Tilings}

\author{Olivier Bodini}
\address{LIP6, CNRS \& Univ. Paris 6,\\ 4 place Jussieu 75005 Paris - France}
\author{Thomas Fernique}
\author{Michael Rao}
\address{Poncelet Lab., CNRS \& Independent Univ. Moscow,\\ 119002, Bolshoy Vlasyevskiy Per. 11, Moscow - Russia}
\author{\'Eric R\'emila}
\address{LIP, CNRS \& ENS de Lyon \& Univ. Lyon,\\ 46 all\'ee d'Italie 69007 Lyon - France}

\begin{abstract}
The rhombus tilings of a simply connected domain of the Euclidean plane are known to form a flip-connected space (a flip is the elementary operation on rhombus tilings which rotates $180^{\circ}$ a hexagon made of three rhombi).
Motivated by the study of a quasicrystal growth model, we are here interested in better understanding how ``tight'' rhombus tiling spaces are flip-connected.
We introduce a lower bound (Hamming-distance) on the minimal number of flips to link two tilings (flip-distance), and we investigate whether it is sharp.
The answer depends on the number $n$ of different edge directions in the tiling: positive for $n=3$ (dimer tilings) or $n=4$ (octogonal tilings), but possibly negative for $n=5$ (decagonal tilings) or greater values of $n$.
A standard proof is provided for the $n=3$ and $n=4$ cases, while the complexity of the $n=5$ case led to a computer-assisted proof (whose main result can however be easily checked by hand).
\end{abstract}

\end{frontmatter}

\section{Introduction}

The discovering in the early 1980's of non-periodic crystals, soon called {\em quasicrystals}, renewed the interest in {\em tilings}.
Indeed, tilings show that a {\em global} property such as non-periodicity can sometimes be enforced by purely {\em local} (hence physically realistic) constraints, namely the way neighbor tiles can match.
The most celebrated example is probably the {\em Penrose tilings}, which are decagonal rhombus tilings of the plane whose non-periodicity can be enforced just by specifying the way rhombi are allowed to fit around a vertex.
Several books have since been written on quasicrystals, and even those focusing on physical properties have a chapter on tilings, with Penrose tilings in the first place (the interested reader can easily found numerous references).\\

Tilings are also widely studied in statistical mechanics, in particular {\em dimer tilings}.
These are tilings of a bounded subset of a planar grid either by dominoes (square grid) or rhombi (triangular grid).
The aim is to understand statistical properties of the set of tilings of a given domain, called a {\em tiling space}.
The interested reader can refer, {\em e.g.}, to the recent survey \cite{kenyon2}.\\

Strongly related to tilings is the notion of {\em flip} (also called {\em phason-flip} or simply {\em phason}).
A flip is the local operation which acts on rhombus tiling by rotating $180^\circ$ a hexagon made of three rhombi\footnote{For dimer tilings on the square grid, a flip rotates $90^\circ$ a square made of two dominoes.}.
Flips may model transformations which occur in real (quasi)crystals \cite{flip}.
The general aim, which inspires this paper, is to understand the flip dynamics on tiling spaces.\\

In particular, the flip dynamics is said to be {\em ergodic} if the more flips are performed on a tiling, the closer it is from a typical tiling (that is, a tiling chosen uniformly at random in the tiling space).
Ergodicity thus provides, when it holds, an algorithmic method to numerically investigate statistical properties of tiling spaces.
Ergodicity have been proven first for dimer tilings \cite{propp} and then for any rhombus tilings \cite{kenyon} (for simply connected domains).\\

However, ergodicity does not tell anything about the number of flips which shall be performed before being relatively close to a typical tiling.
This is not only a drawback from a computational viewpoint (how far a simulation shall be run?), but can also lead to physical mispredictions: ergodicity can yield properties that have no chance to be experimentally observed because the flip dynamics is much too slow.
Ergodicity results can be refined by studying the {\em mixing time} of the Markov process whose steps consist in performing a flip on a vertex chosen uniformly at random (when possible).
A polynomial bound on the mixing time of dimer tilings have been obtained in \cite{mixing2}, relying on the analysis of a related Markov chain \cite{mixing1,mixing3}.
Numerical experiments suggest that a similar bound holds for rhombus tilings whose rhombi are defined over four possible edge directions (octogonal tilings) \cite{destainville2}.\\

In this context, the motivation of this paper comes from the model of quasicrystal growth introduced in \cite{aperiodic}.
Starting from a typical tiling (assumed to be an entropically-stabilized quasicrystal at high temperature), flips are randomly performed, with the probability of each flip depending on its neighborhood (assumed to model the energy of the quasicrystal).
This process can be seen as the minimization of the energy of a tiling by a stochastic gradient descent in the tiling space.
The question is whether a minimal energy tiling is reached (a sort of partial ergodicity), and at which rate (mixing time).
In other word, the question is whether the shortest paths (in terms of flips) of tiling spaces are enough ``regular'' to be easily found by the gradient descent.
As an attempt to measure this regularity, this paper introduces a distance, called Hamming-distance, and compares it to the above flip-distance.\\

The main result obtained in this paper is that rhombus tiling spaces are regular when rhombi are defined over at most four edge directions (dimer or octogonal tilings), but loose this regularity for five directions (decagonal tilings) or more.
Theorem~\ref{th:main} formally states this in Section~\ref{sec:comparison}.
Before, Section~\ref{sec:tilings} and \ref{sec:distances} first introduce main definitions, especially flip- and Hamming-distances.\\

Let us also mention that flips can be defined on tilings of the space by rhomboedra\footnote{In this case, flips act on dodecahedra made of four rhomboedral tiles.}.
Ergodicity does not always hold in such cases \cite{destainville}, and our results can be extended to show that rhomboedra tiling spaces are not regular, except when there is only one rhombohedron up to isometry (what is already known), and maybe for a finite number of small domains.

\section{Rhombus tilings and pseudoline bundles}
\label{sec:tilings}

Let $\vec{v}_1,\ldots,\vec{v}_n$ be pairwise non-collinear unit vectors of the Euclidean plane.
These $n$ vectors define $\binom{n}{2}$ {\em rhombus tiles} $\{\lambda\vec{v}_i+\mu\vec{v}_j~|~0\leq \lambda,\mu\leq 1\}$.
A {\em domain} $D$ is a connected subset of the Euclidean plane.
A {\em rhombus tiling} of $D$ is a set of translated rhombus tiles which cover $D$, with the intersection of two tiles being either empty, or a vertex, or an entire edge.
We here focus on {\em zonotopal} domains of the form
$$
D=\{\lambda_1\vec{v}_1+\ldots+\lambda_n\vec{v}_n~|~\forall k,~0\leq \lambda_k\leq a_k\},
$$
where $a_1,\ldots, a_n$ are positive integers.
One simply writes $D=(a_1,\ldots,a_n)$.\\

\begin{figure}[hbtp]
\centering
\includegraphics[width=0.9\textwidth]{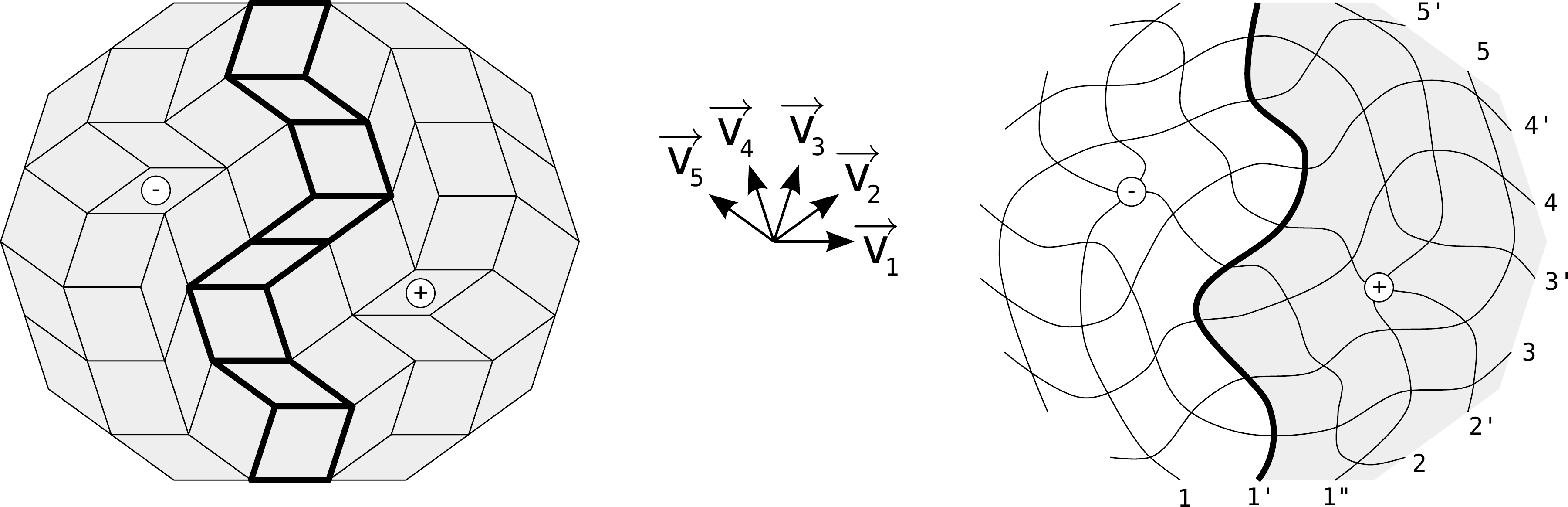}
\caption{A rhombus tiling of the zonotope $(3,2,2,2,2)$ with an emphasized $\vec{v}_1$-directed ribbon (left).
The corresponding numbered pseudolines arrangement (right, with $1'$ corresponding to the emphasized ribbon on the left).
We marked two vertices (right, and the corresponding tiles on the left) by a sign indicating whether they are in $i'_+$ or $i'_-$ (see below).}
\label{fig:rhombus_tiling_pseudolines}
\end{figure}

While studying Penrose tilings, de Bruijn introduced ``ribbons'' of tiles, now generally called {\em de Bruijn lines} \cite{debruijn}.
Formally, a {\em $\vec{v}_i$-directed ribbon} is an inclusion-maximal sequence of tiles, with each one being adjacent to the following one by an edge parallel to $\vec{v}_i$ (Fig.~\ref{fig:rhombus_tiling_pseudolines}, left).\\

Drawing through a $\vec{v}_i$-directed ribbon a continuously differentiable curve $i$ with no tangent parallel to $\vec{v}_i$ splits the domain in two connected components $Si_+:=\{i+\lambda\vec{v}_i~|~\lambda>0\}$ and $i_-:=\{i-\lambda\vec{v}_i~|~\lambda>0\}$ (Fig.~\ref{fig:rhombus_tiling_pseudolines}, right).
Such a curve is called a pseudoline.
The union of pseudolines drawn on the ribbons of a tiling yields a kind of {\em pseudoline arrangement}, where the properties of classical pseudoline arrangements (see, {\em e.g.}, \cite{matroids})
\begin{enumerate}
\item two pseudolines cross in {\em exactly} one point;
\item the intersection of {\em all} pseudolines is empty,
\end{enumerate}
shall here be modified into
\begin{enumerate}
\item two pseudolines cross in {\em at most} one point;
\item the intersection of {\em any three} pseudolines is empty.
\end{enumerate}
We call {\em $i$-th bundle} the set of pseudolines drawn in $\vec{v}_i$-directed ribbons and numbered according to their crossings with the line directed by $\vec{v}_i$.
We denote by $i,i',\ldots,i^{(a_i)}$ the $a_i$ pseudolines of the $i$-th bundle (Fig.~\ref{fig:rhombus_tiling_pseudolines}, right).

\section{Flip- and Hamming-distances}
\label{sec:distances}

Whenever three tiles of a rhombus tilings form a hexagon, a $180^{\circ}$ rotation around their common vertex yields a new rhombus tiling of the same domain; such a local operation is called {\em flip} (Fig.~\ref{fig:flip}, left).
In terms of pseudolines, a flip continuously deforms pseudolines so that exactly one crossing of two pseudolines is moved across exactly one pseudoline (Fig.~\ref{fig:flip}, right).
In other words, a flip {\em inverts} an inclusion-minimal {\em triangle} (formal definition below).\\

\begin{figure}[hbtp]
\centering
\includegraphics[width=0.9\textwidth]{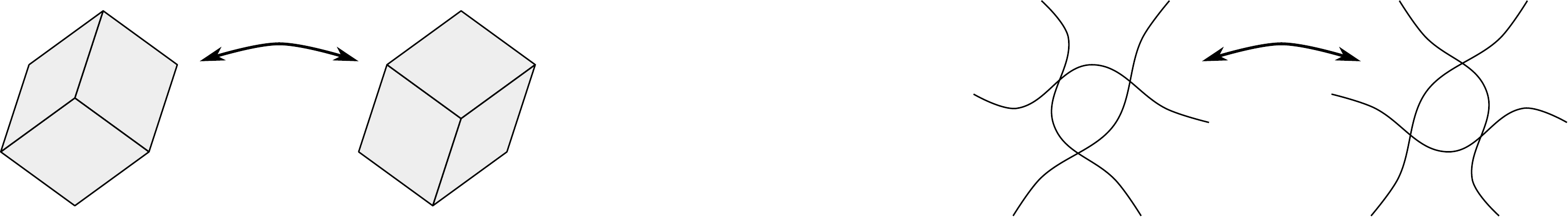}
\caption{A flip on rhombus tiles (left) and on pseudolines (right).}
\label{fig:flip}
\end{figure}

The flip-invariance of a domain leads to define its {\em rhombus tiling space}: this is the undirected graph whose vertices are the rhombus tilings of the domain, with two tilings being connected by an edge if they differ by a flip.
This space is naturally endowed with the following metric.

\begin{definition}
The {\em flip-distance} between two tilings of a rhombus tiling space is the minimal number of flips to perform for transforming one into the other.
\end{definition}

Kenyon \cite{kenyon} proved that the rhombus tiling space of a finite simply connected domain ({\em e.g.}, a zonotope) is connected, that is, the flip-distance is bounded on such a space.
This yields a first insight on the structure of rhombus tiling spaces.
In order to further understand this structure, we introduce a second metric, which relies on the notion of {\em signed triangle}.\\

\begin{figure}[hbtp]
\centering
\includegraphics[width=0.9\textwidth]{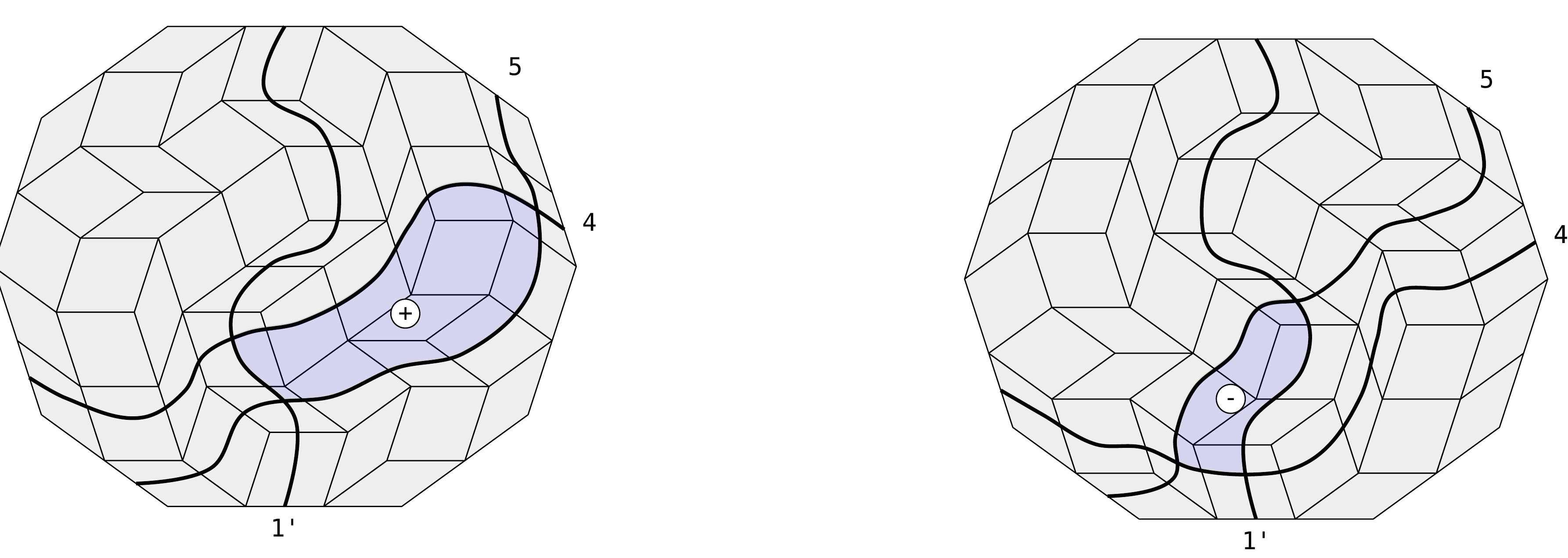}
\caption{Two tilings in which the triangle $(1',4,5)$ has a different sign.}
\label{fig:triangles}
\end{figure}

\begin{definition}
A {\em triangle} is a triple $(i^{(a)},j^{(b)},k^{(c)})_{i<j<k}$ of pairwise crossing pseudolines.
It is said to be {\em positive} if $i^{(a)}\cap j^{(b)}\in k^{(c)}_+$, {\em negative} otherwise.
\end{definition}

Since a pseudoline drawn in a ribbon of a tiling also appears in any other tiling of the same domain, the notion of triangle depends not on a particular tiling but on a rhombus tiling space.
However, since the representation of a pseudoline depends on the particular tiling it is drawn on, so does also the representation of a triangle.
In particular, a triangle can have a different sign in two tilings of the same space: it is said to be {\em inverted} with respect to these two tilings (Fig.~\ref{fig:triangles}).

\begin{definition}
The {\em Hamming-distance} between two tilings of a rhombus tiling space is the number of triangles which have different sign in these two tilings.
\end{definition}

If we associate with each tiling of a rhombus tiling space a string whose $k$-th symbol is the sign in this tiling of the $k$-th triangle\footnote{We can, for example, order triangles by the lexicographic order induced on triples of pseudolines by the ordering of pseudolines.}, then the Hamming-distance between two tilings is exactly the usual Hamming distance between their associated strings.
One can moreover show that such a string completly characterizes the tiling it is associated with.\\

The Hamming-distance provides a natural and simple lower bound on the flip-distance.
Indeed, since each flip changes the sign of exactly one triangle, transforming a tiling into another one requires at least as many flips as there are triangles with a different sign in the two tilings.
The main question that we examine in the next section is: are these distances actually equal?\\

Intuitively, an equality between flip- and Hamming-distances would mean that the rhombus tiling space is rather friendly: tilings can be efficiently connected by flips.
On the contrary, an inequality would mean that, even if connected, the rhombus tiling space has obstructions: ``auxiliary'' flips must sometimes be back and forth performed in order to free ``efficient'' flips.\\

\begin{figure}[hbtp]
\centering
\includegraphics[width=0.85\textwidth]{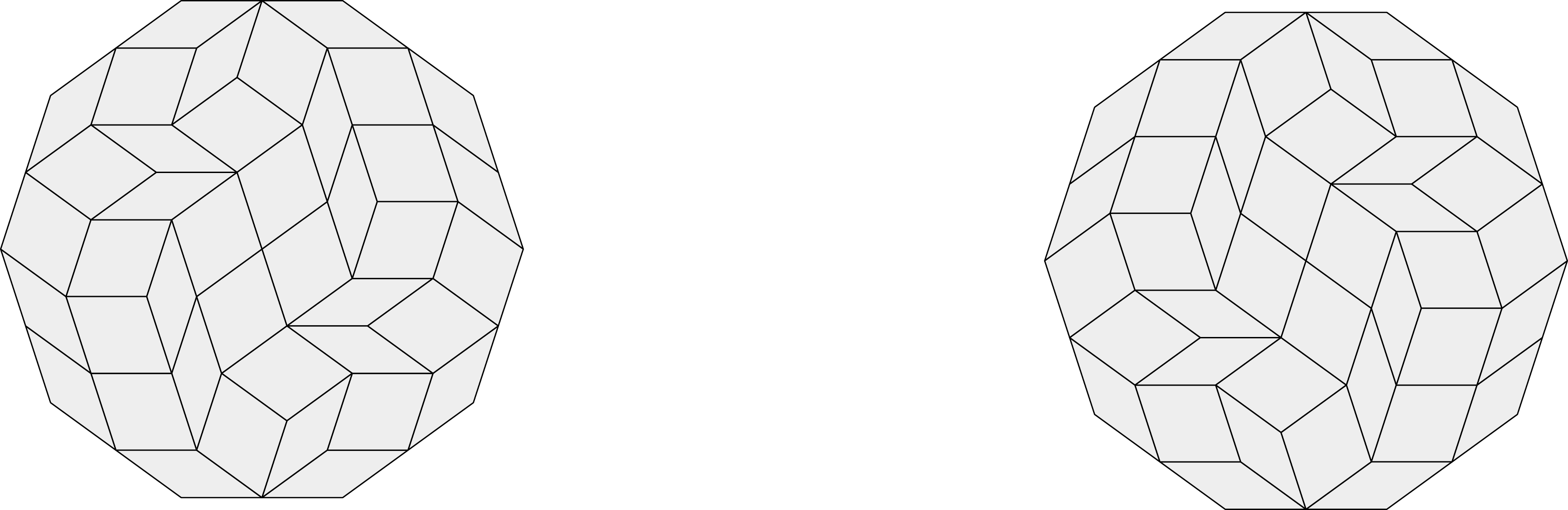}
\caption{
Two of the \numprint{16832230} tilings of the zonotope $(2,2,2,2,2)$.
One checks that $24$ out of $80$ triangles have a different sign in both tilings, whose Hamming-distance is thus $24$.
Theorem~\ref{th:main} (below) ensures that they are also at flip-distance $24$.
}
\label{fig:hamming}
\end{figure}

Note also that computing the flip-distance between two tilings {\em a priori} requires to compute the flip-distance to one of them to all the tilings of the space, up to find the other one.
This is however highly unefficient, because the size of a rhombus tiling space is known to be at least exponential in the size ({\em i.e.}, number of tiles) of its tilings.
On the contrary, computing the Hamming-distance between two tilings can be done very efficiently (in sesquilinear time in the size of the tilings).
The equality of flip- and Hamming distance has thus also important implications from a computational viewpoint (Fig.~\ref{fig:hamming}).

\section{Distance comparison}
\label{sec:comparison}

\subsection{Main result}

\begin{theorem}\label{th:main}
Consider the rhombus tiling space of a zonotopal domain $D$.
The flip- and Hamming-distances are equal on this tiling space if and only if
$$
D\in\{(a,b,c),~(a,b,c,d),~(a,b,c,d,1),~(2,2,2,2,2)~|~a,b,c,d\geq 1\}.
$$
\end{theorem}

This theorem is established by several lemmas which depend on the number or size of bundles of pseudoline arrangements associated with rhombus tiling spaces (next subsections), as well as on the following general lemma:

\begin{lemma}\label{lem:main}
The flip- and Hamming-distances are equal on a rhombus tiling space if, for any two distinct tilings, each tiling admits an inclusion-minimal triangle with a different sign in the other tiling.
\end{lemma}

\begin{proof}
Consider two tilings at Hamming-distance $n>0$.
By hypothesis, the first one admits an inclusion-minimal triangle with a different sign in the other one.
This inclusion-minimal triangle can be inverted by a flip.
This flip changes the sign of the triangle without affecting any other triangle, hence decreases by one the Hamming-distance between both tilings.
Iterating this $n$ times yields a sequence of $n$ flips which reduces the Hamming-distance between tilings to zero, that is, which connect them.
This thus bounds by below the flip-distance between the two original tilings.
The claimed result follows, since one already know that the Hamming-distance is lesser than or equal to the flip-distance.
\end{proof}

\subsection{Three bundles}

We focus here on the tilings of zonotopes in $(a,b,c)_{a,b,c\geq 1}$, that is, three-bundle rhombus tiling spaces.

\begin{lemma}\label{lem:config1}
Consider a triangle $(S_1,S_2,S_3)$ whose $S_1$- and $S_2$-sides are cut by a pseudoline $S_4$.
If $(S_1,S_2,S_3)$ has a different sign in two tilings, then so does $(S_2,S_3,S_4)$ (in the same two tilings).
\end{lemma}

\begin{proof}
First, since $S_1$, $S_2$ and $S_3$ (resp. $S_1$, $S_2$ and $S_4$) pairwise cross, they are in different bundles; this ensures that $S_3$ and $S_4$ are in the same bundle because there are only three bundles.
Now, since $(S_1,S_2,S_3)$ has a different sign in two tilings, any sequence of flips transforming one tiling into the other must move $S_1\cap S_2$ across $S_3$.
Since $S_3$ cannot be moved across $S_4$ (two pseudolines in the same bundle cannot cross), one first needs to move $S_1\cap S_2$ across $S_4$, what changes the sign of $(S_1,S_2,S_4)$.
\end{proof}

\begin{figure}[hbtp]
\centering
\includegraphics[width=0.9\textwidth]{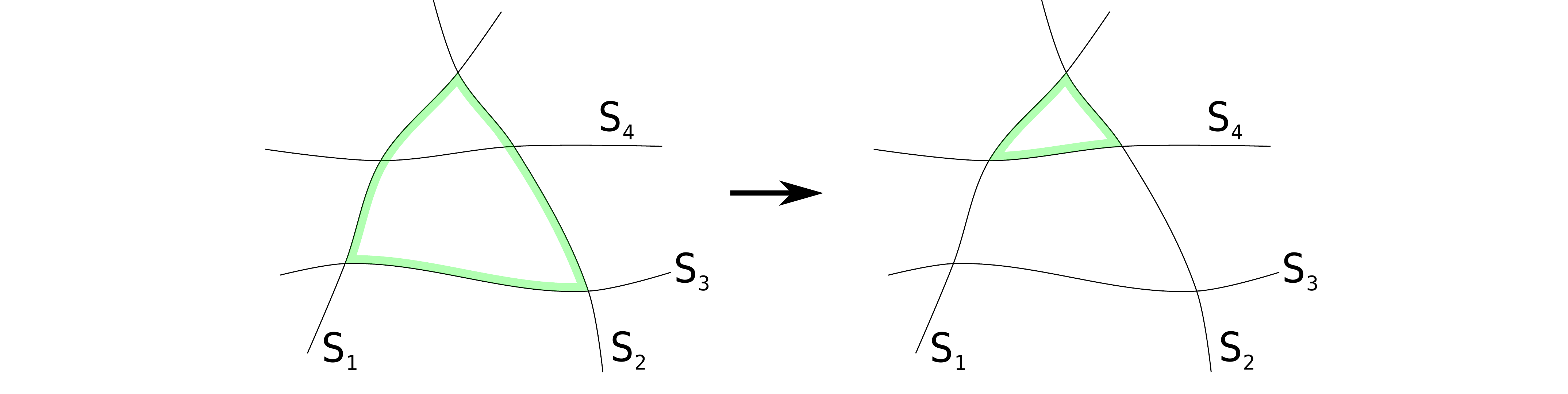}
\caption{If $(S_1,S_2,S_3)$ has a different sign in two tilings, then so does $(S_1,S_2,S_4)$ (Lemma~\ref{lem:config1}).}
\label{fig:config1}
\end{figure}

\begin{lemma}\label{lem:three_bundles}
For any two tilings of a zonotope $(a,b,c)_{a,b,c\geq 1}$, each tiling admits an inclusion-minimal triangle with different sign in the other tiling.
\end{lemma}

\begin{proof}
Since the tilings are distinct, they gives a different sign to some triangle.
If such a triangle is not inclusion-minimal, then it is cut by a pseudoline.
The previous lemma yields a new triangle, with different sign in the two tilings, included in the first triangle.
This can be iterated up to get an inclusion-minimal triangle.
\end{proof}

The two previous lemmas, together with Lemma~\ref{lem:main}, yield the first case of Theorem~\ref{th:main}.
Actually, in this three-bundle case, there is a much better result on the structure of the rhombus tiling space.
Indeed, this case corresponds to the so-called {\em dimer tilings}, where the tiling space is known to have a structure of {\em distributive lattice} \cite{propp}.

\subsection{Four bundles}

We focus here on the tilings of zonotopes $(a,b,c,d)_{a,b,c,d\geq 1}$, that is, four-bundle rhombus tiling spaces.

\begin{lemma}\label{lem:config2}
Consider a triangle $(S_1,S_2,S_3)$ whose $S_1$- and $S_2$-sides are cut by a pseudoline $S_4$, which also cuts $S_3$.
If $(S_1,S_2,S_3)$ has a different sign in two tilings, then so does $(S_2,S_3,S_4)$ or both $(S_1,S_3,S_4)$ and $(S_1,S_2,S_4)$.
\end{lemma}

\begin{proof}
To change the sign of $(S_1,S_2,S_3)$, we need to make it inclusion-minimal (in order to perform a flip which moves $S_1\cap S_2$ across $S_3$), that is, to continuously deform pseudolines so that $S_4$ does not any more intersect $(S_1,S_2,S_3)$.
This can be done in two way: either in one step by moving $S_2\cap S_3$ across $S_4$, or in two steps, by moving, first, $S_1\cap S_4$ across $S_3$ (so that $S_4$ crosses $S_1$- and $S_2$-sides of $(S_1,S_2,S_3)$) and, second, $S_1\cap S_2$ across $S_4$.
The first way changes the sign of $(S_2,S_3,S_4)$, while the second way changes the sign of both $(S_1,S_3,S_4)$ and $(S_1,S_2,S_4)$.
\end{proof}

\begin{figure}[hbtp]
\centering
\includegraphics[width=0.9\textwidth]{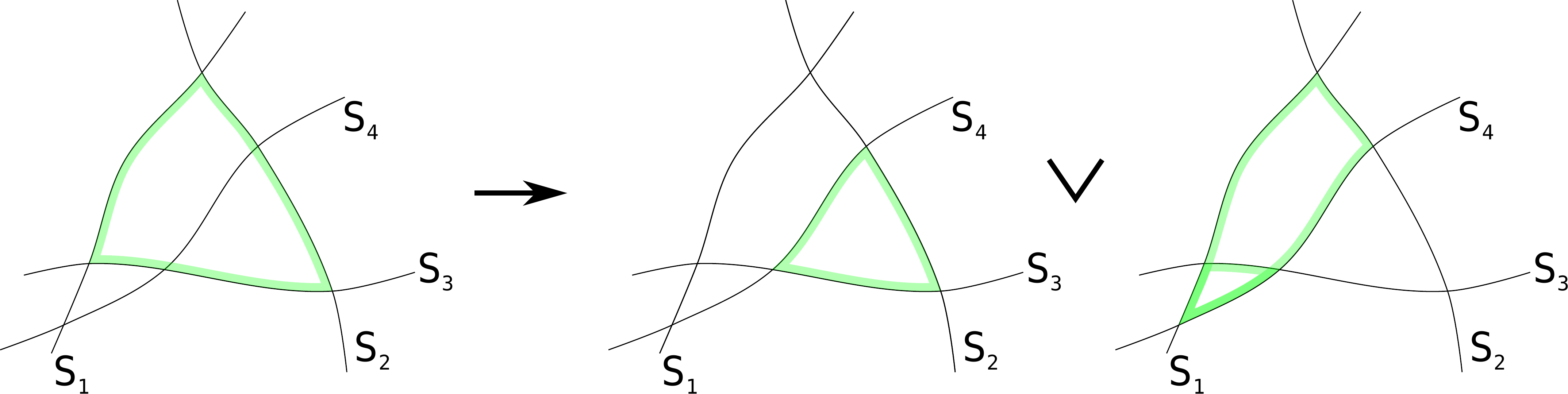}
\caption{If $(S_1,S_2,S_3)$ has a different sign in two tilings, then so does $(S_2,S_3,S_4)$ or both $(S_1,S_3,S_4)$ and $(S_1,S_2,S_4)$ (Lemma~\ref{lem:config2}).}
\label{fig:config2}
\end{figure}

\begin{lemma}\label{lem:four_bundles}
For any two tilings of a zonotope $(a,b,c,d)_{a,b,c,d\geq 1}$, each tiling admits an inclusion-minimal triangle with a different sign in the other tiling.
\end{lemma}

\begin{proof}
Let $(S_1,S_2,S_3)$ be a inclusion-minimal triangle among the {\em inverted} triangles, that is, those with a different sign is the two tilings.
It thus cannot be cut a pseudoline in the bundles of $S_1$, $S_2$ or $S_3$.
If it is not cut by a pseudoline, then the lemma is proven.\\

Otherwise, let $S_4$ be a pseudoline which cuts $(S_1,S_2,S_3)$, say on its $S_1$- and $S_2$-sides.
It also crosses $S_3$ (otherwise $(S_1,S_2,S_4)$ would be inverted by Lemma~\ref{lem:config1}, and included in $(S_1,S_2,S_3)$, contradicting its minimality).
Moreover, we choose $S_4$ so that, in its bundle, it the closest pseudoline to $S_1\cap S_3$, that is, no other pseudoline of its bundle cuts $(S_1,S_3,S_4)$.
By minimality of $(S_1,S_2,S_3)$, $(S_1,S_2,S_4)$ is not inverted.
Lemma~\ref{lem:config2} thus ensures that both $(S_1,S_3,S_4)$ and $(S_2,S_3,S_4)$ are inverted.
If $(S_1,S_3,S_4)$ is inclusion-minimal, then the lemma is proven (Fig.~\ref{fig:four_bundles_b}).\\

Otherwise, $(S_1,S_3,S_4)$ is cut by a pseudoline.
This pseudoline cannot be in the bundle of $S_4$ (because of the way $S_4$ was chosen).
A pseudoline in the bundle of $S_3$ cannot cross the $S_1$-side of $(S_1,S_3,S_4)$, as a subset of the $S_1$-side of $(S_1,S_2,S_3)$.
It can thus only cross the $S_4$-side.
But this is impossible because two pseudolines cross at most once (thus the pseudoline cannot both enter and exit $(S_1,S_3,S_4)$).
It is thus in the bundles of either $S_1$ or $S_2$.
Assume that there is such a pseudoline in the bundle of $S_1$, and take the closest one to $S_3\cap S_4$, say $S_1'$.
By Lemma~\ref{lem:config1}, $(S_1',S_3,S_4)$ is inverted.
If it is not cut by a pseudoline (which is necessarily in the bundle of $S_2$), then the lemma is proven (Fig.~\ref{fig:four_bundles_c}).\\

Otherwise, let $S_2'$ be the pseudoline in the bundle of $S_2$ which is the closest to $S_3\cap S_4$.
Since $(S_2,S_3,S_4)$ is inverted, Lemma~\ref{lem:config1} ensures that $(S_2',S_3,S_4)$ is inverted.
If it is not cut by a pseudoline, then the lemma is proven (Fig.~\ref{fig:four_bundles_d}).\\

Otherwise, consider a pseudoline which cuts $(S_2',S_3,S_4)$.
It cannot be in the bundle of $S_2$, $S_3$ or $S_4$, because of the way they were chosen (for the bundle of $S_3$ this is a bit more delicate: remark that a pseudoline in this bundle which would cut $(S_2,S_3,S_4)$ would also necessarily cut $S_4$, entering in this way $(S_1,S_3,S_4)$, but it then cannot exit $(S_1,S_3,S_4)$: crossing twice $S_4$ is forbidden, crossing $S_3$ also, and the $S_1$-side of $(S_1,S_3,S_4)$ cannot be crossed by any pseudoline).
Hence, a pseudoline which cuts $(S_2,S_3,S_4)$ is in the bundle of $S_1$.
More precisely, it is necessarily $S_1'$ itself, because of the way it was chosen.
Since $(S_2',S_3,S_4)$ is inverted, Lemma~\ref{lem:config2} ensures that at least one of $(S_1',S_2',S_3)$ and $(S_1',S_2',S_4)$ is inverted.
But the successive choices of $(S_1,S_2,S_3)$, $S_4$, $S_1'$ and $S_2'$ ensures that no more pseudoline can cut $(S_1',S_2',S_3)$ or $(S_1',S_2',S_4)$.
Thus, at least one of these two triangles is inclusion-minimal and inverted: the lemma is proven (Fig.~\ref{fig:four_bundles_e}).
\end{proof}

\begin{figure}[hbt]
\centering
\subfloat[$(S_1,S_3,S_4)$ is inverted]{\label{fig:four_bundles_b}\includegraphics[width=0.48\textwidth]{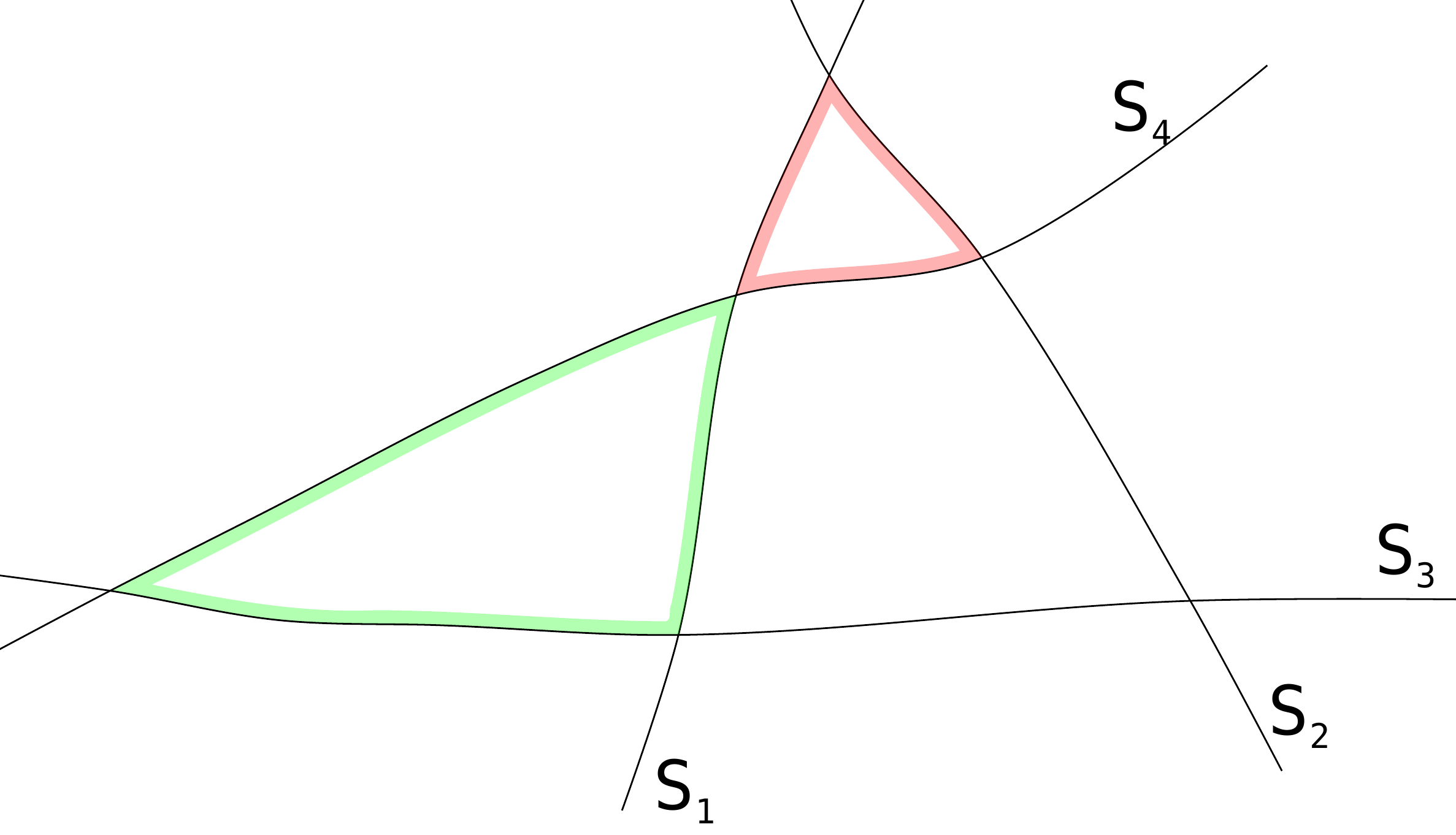}}
\hfill
\subfloat[$(S_1',S_3,S_4)$ is inverted]{\label{fig:four_bundles_c}\includegraphics[width=0.48\textwidth]{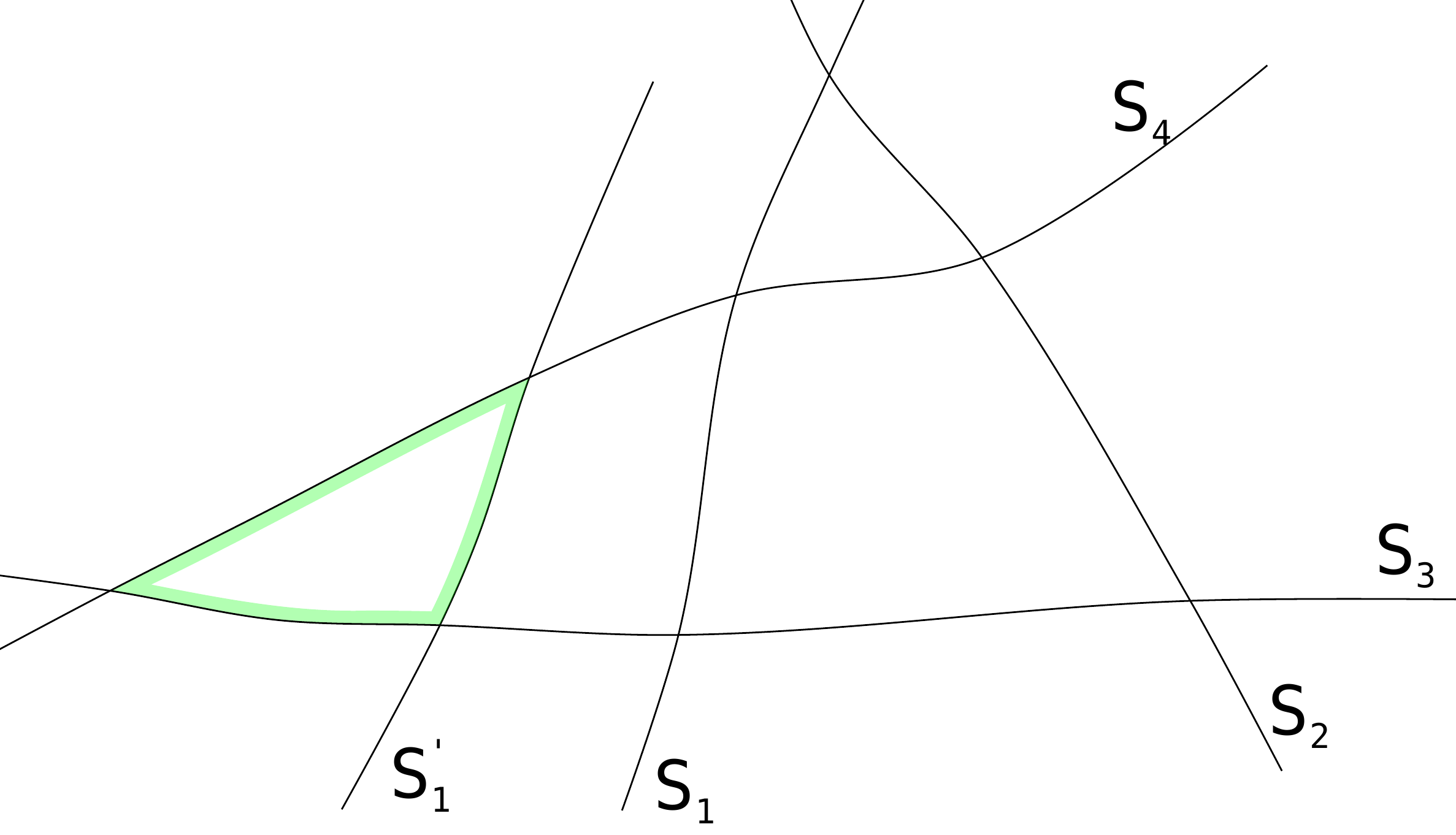}}\\
\subfloat[$(S_2',S_3,S_4)$ is inverted]{\label{fig:four_bundles_d}\includegraphics[width=0.48\textwidth]{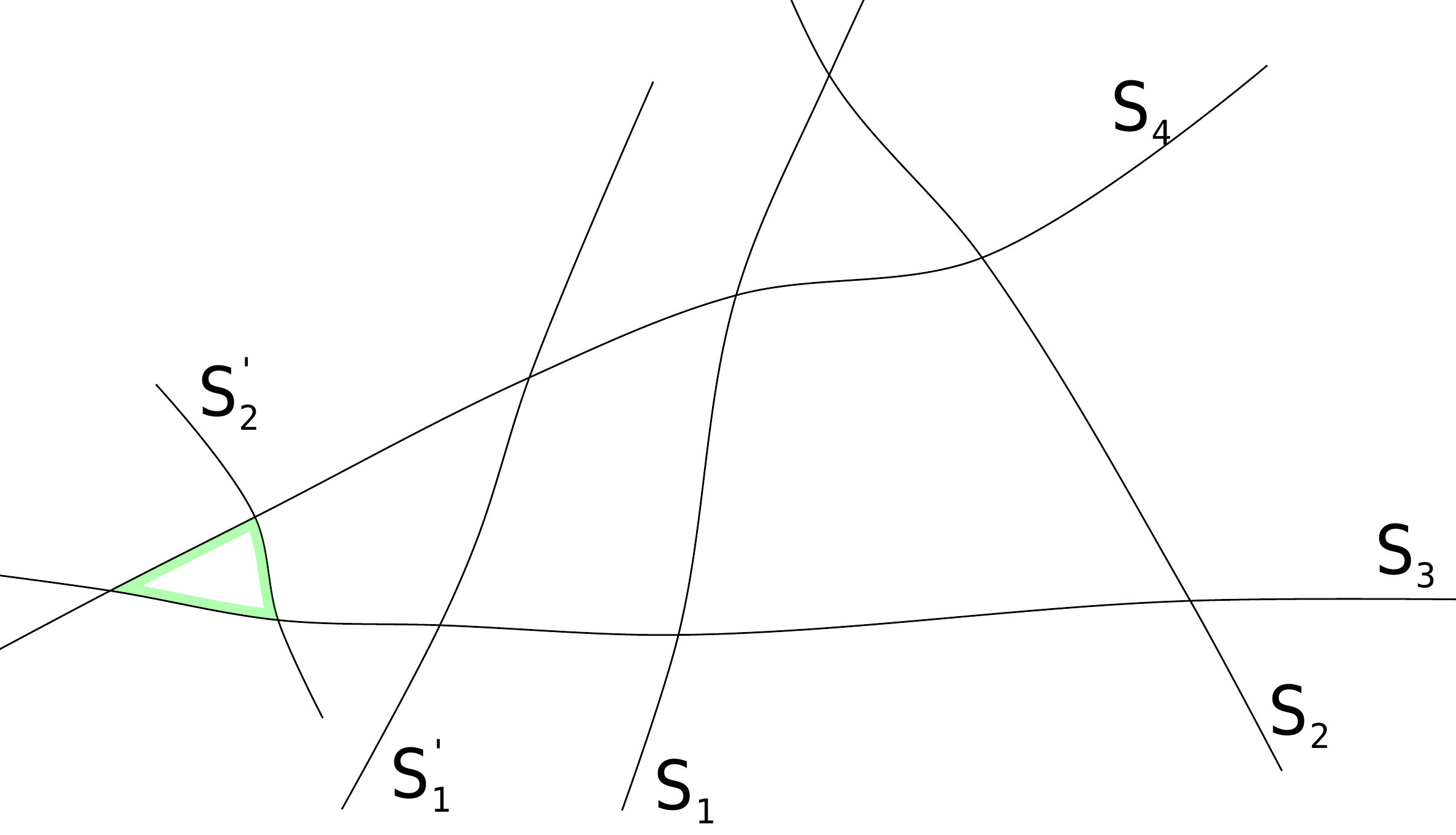}}
\hfill
\subfloat[$(S_1',S_2',S_3)$ or $(S_1',S_2',S_4)$ is inverted]{\label{fig:four_bundles_e}\includegraphics[width=0.48\textwidth]{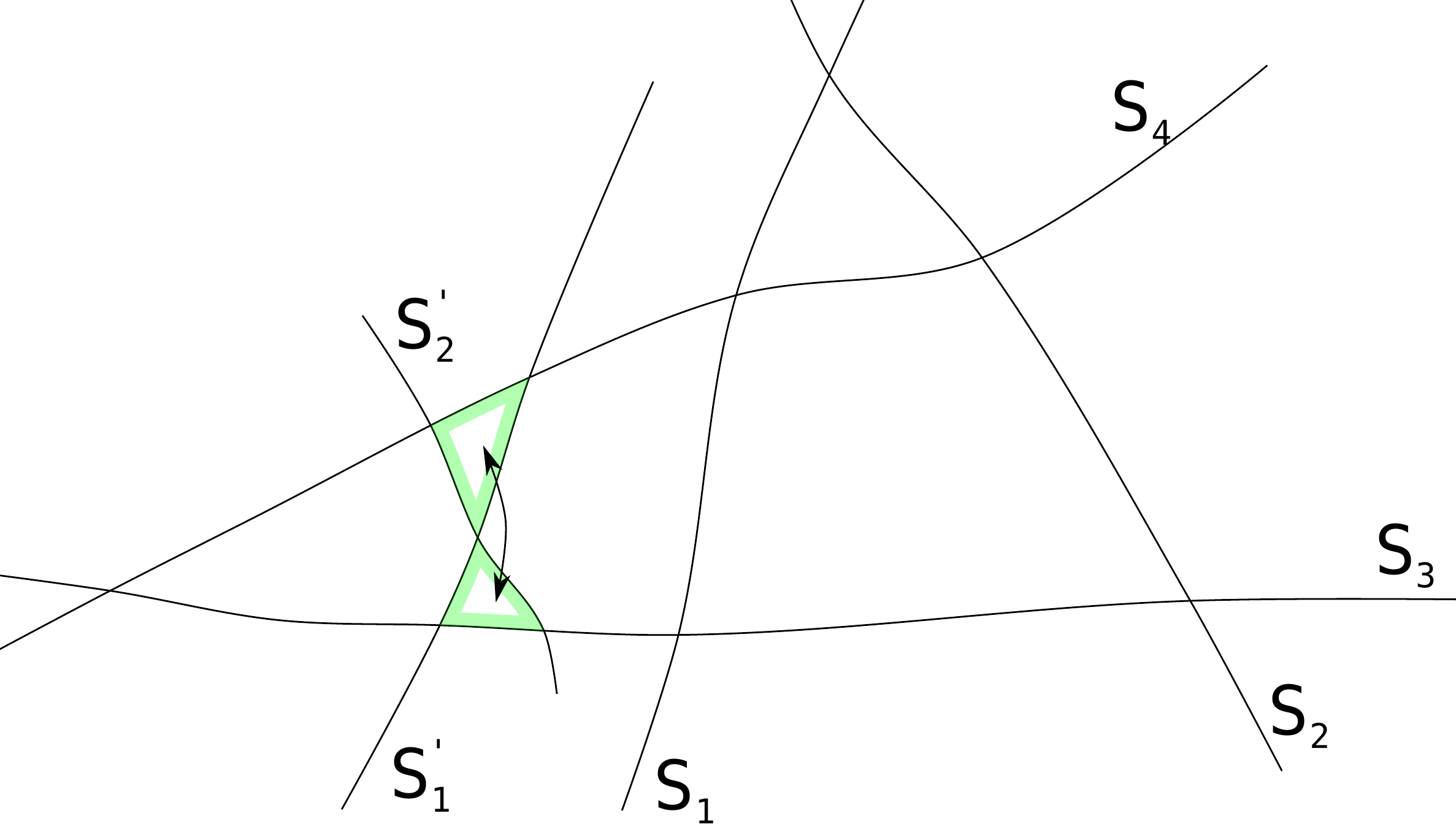}}
\caption{The case study of Lemma~\ref{lem:four_bundles}.}
\label{fig:four_bundles}
\end{figure}

The two previous lemmas, together with Lemma~\ref{lem:main}, yield the first case of Theorem~\ref{th:main}.
On the contrary to three-bundle rhombus tiling spaces, four-bundle ones do not admit a structure of distributive lattice, but only of graded poset \cite{CR}.
The result here obtained thus really yields a deeper insight on the structure of rhombus tiling spaces.

\subsection{Five bundles}

We focus here on the tilings of zonotopes $(a,b,c,d,e)_{a,b,c,d,e\geq 1}$, that is, five-bundle rhombus tiling spaces.
In the conference version of this paper \cite{dgci}, we conjectured that equality of flip- Hamming-distances holds in this case\footnote{Five-bundle tilings were there called $5\to 2$ tilings -- an equivalent terminology.}.
However, mimicking the proof of Lemma~\ref{lem:four_bundles} led to a huge case study that we were not able to carry out.
We therefore decided to write a program to carry out this case study.\\

The general principle of this program is the following.
It starts with a triangle which is assumed to be inverted w.r.t. two tilings of the space (such a triangle exists for any two different tilings).
It then considers all the way an additional pseudoline can cut this triangle, and recursively searches for either an inverted triangle which cannot any more be cut by any pseudoline, or a configuration with none of its inclusion-minimal triangles inverted.
Lemma~\ref{lem:main} then ensures that flip- and Hamming distances are equal in the former case, or different in the latter case.\\

Let us here describe precisely how the program works, since it is not that straightforward.
First, let us define the following rooted tree:
\begin{enumerate}
\item each node corresponds to a pseudoline arrangement, where triangles are labelled ``inverted'' or ``non-inverted''.
\item nodes can be of type \texttt{AND} or \texttt{OR}, with nodes \texttt{AND} having a distinguished ``inverted'' inclusion-minimal triangle;
\item the root is a node \texttt{AND} corresponding to a single (distinguished) triangle;
\item the children of a node \texttt{AND} are nodes \texttt{OR}, one for each way a pseudoline can be added which cut the distinguished triangle and the created triangles can be consistently\footnote{The labelling of a pseudoline arrangement is consistent if there exists another pseudoline arrangement such that the triangles with a different sign in both arrangements are exactly the one labelled ``inverted''.} labelled ``inverted'' or ``non-inverted'';
\item the children of a node \texttt{OR} are nodes \texttt{AND}, one for each ``inverted'' inclusion-minimal triangle which do not appear in the parent node \texttt{AND}, with this ``inverted'' triangle being distinguished.
\end{enumerate}
We then {\em evaluate} this tree as follows: leaves \texttt{AND} and \texttt{OR} are respectively eva\-lua\-ted to \texttt{TRUE} and \texttt{FALSE}, while nodes \texttt{AND} and \texttt{OR} are respectively evaluated to the conjunction and the disjunction of the evaluations of their children.
It is not hard to see that if the root is evaluated to \texttt{TRUE}, then any pair of different tilings shall admit an inverted minimal-inclusion triangle.\\

The above defined tree has however no reason to be finite (remind that we are not considering a particular zonotope, but an infinite family, that is, there is a restriction of the number of pseudoline bundles, but not on the number of pseudolines).
It nevertheless may admit a subtree with the same root, such that if the root of this subtree is evaluated to \texttt{TRUE}, then the root of the general tree can only be evaluated to \texttt{TRUE}.
Indeed, whenever a child of a node \texttt{OR} is evaluated to \texttt{TRUE}, the evaluation of the other children does not matter: this allows to ignore large parts of the general tree, hopefully enough to obtain a finite subtree.
Such a subtree, in the case its root is evaluated to \texttt{TRUE}, provides a proof that any pair of different tilings admit an inverted minimal-inclusion triangle (hence that flip- and Hamming-distances are equal, by Lemma~\ref{lem:main}): we call it a {\em proof-subtree} (there may be several ones).\\

The above approach can only prove equality since it focus on the evaluation to \texttt{TRUE} of the root.
However, the leaves \texttt{OR} discovered while searching a proof-subtree are good counterexample candidates.
Indeed, such a leaf corresponds to a pseudoline arrangement where, at least locally, no inclusion-minimal triangle is inverted.
If this also holds through the whole pseudoline arrangement, then Lemma~\ref{lem:main} yields that flip- and Hamming-distances are not equal.
Thus, while searching for a proof-subtree, our program also outputs leaves \texttt{OR}, and we separatly check whether they are counterexample or not.\\

It remains to explain the strategy of our program for skipping large parts of the general tree while exploring it.
There are two main principles.\\

The first principle is that, when adding a pseudoline $S$ which cuts a distinguished triangle, we systematically assume that this pseudoline is, among the pseudolines in the same bundle which cross the same edges of the triangle, the closest one to one of the three vertex of the triangle (this is exactly what we did several times in the proof of Lemma~\ref{lem:four_bundles}, actually, the program just mimics the proof in the more complicated five-bundle case).
This allows to endow sections of the triangle edges with a label specifying that they cannot be cut by a pseudoline of this bundle.
While recursively adding new pseudolines during the tree exploration, we add more and more such labels: this is precisely what leads to nodes where no pseudoline crossing the distinguished triangle can be added, that is, leaves \texttt{AND} evaluated to \texttt{TRUE}.\\

The second principle is to combine depth-first and breadth-first searches according to a well-chosen trade-off.
On the one hand, since we just want to find, for each node \texttt{OR}, a child whose evaluation is \texttt{TRUE}, it seems better to perform a breadth-first search that can be stopped as soon as such a child is found.
But this breadth-first search is untractable in practice, because this leads to examine too much cases.
However, since these cases are ``more or less promising'' (depending on how looks the pseudoline arrangement and in which part of it an inverted inclusion-minimal triangle is searched), it is worth performing a depth-first search on the most ``promising'' ones.
Of course, it is very hard to estimate how much a given case is ``promising''.
We therefore proceed as follow: we associate a {\em weight} with each node whose children have not all been explored -- the most ``promising'' seems a node, according to some ``homemade'' heuristics, the heavier it is -- and we perform each step of the search on the heavier node.
Weights are updated at each step.\\

We shall also mention that, when searching for children of a node \texttt{AND}, the program does not compute exactly the consistent labelling (this is very time-consuming).
Instead of this, it considers all the labelling which are consistent w.r.t. Lemmas \ref{lem:config1} and \ref{lem:config2}.
This leads to possibly consider some extra labelling (which are not detected as inconsistent), but this can only prevent the corresponding \texttt{AND} node to be evaluated to \texttt{TRUE}, that is, we can miss a proof-subtree but not obtain a false proof-subtree.\\

This concludes our description of the program (which however contains several other tricky optimizations -- the interested reader shall contact Michael Rao for more details).
Let us now discuss the results that we obtained.\\

The main result is rather disappointing: this is a counterexample to our conjecture.
The program indeed found\footnote{That is, while running indefinitely, the program output a leaf \texttt{OR} corresponding to a tiling of this pair, we then search for the tiling such that the triangles inverted w.r.t. these two tilings are exactly those labelled ``inverted'' in the leaf \texttt{OR}, and we check that no such triangle is inclusion-minimal.} a {\em deficient pair} (that is, a pair of tilings whose flip- and Hamming-distances differ) in the zonotope $(3,2,2,2,2)$, namely the pair depicted on Fig.~\ref{fig:deficient_5}.
Let us stress that, once a pair is found, it can be easily checked that it is a counterexample: it indeed suffices to check that each inclusion-minimal triangle of one of the tilings has the same sign in the other tiling (this can be valuable for readers who would dubiously regard computer-assisted proofs).\\

\begin{figure}[hbtp]
\centering
\includegraphics[width=0.9\textwidth]{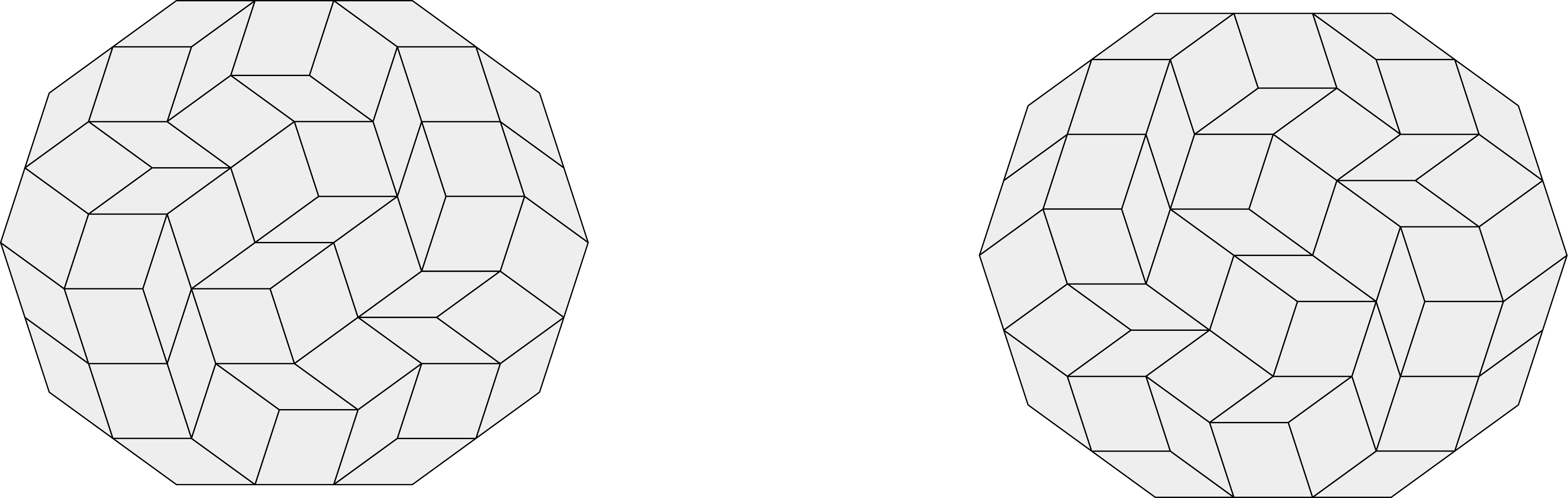}
\caption{A deficient pair of the zonotope $(3,2,2,2,2)$: these tilings are at Hamming-distance $32$ but at flip-distance $34$. This deficiency can be checked by hand: the $10$ inclusion-minimal triangle of one tiling have indeed the same sign in the other tiling.}
\label{fig:deficient_5}
\end{figure}

The program also found several other deficient pairs in this zonotope ($9$ at all, up to symmetries), but there are possibly other ones (that the program can have ignored while searching for a proof-subtree).
An exhaustive checking of the \numprint{139106980443312324} pairs of this tiling space shall give all the deficient pairs, but it is beyond our computational resources.
Let us mention that the flip- and Hamming-distances of all the found pairs differ only by $2$.\\

Since the above counterexample can be easily extended to larger zonotope, the only remaining cases, among five-bundle rhombus tiling spaces, are those associated with the zonotopes $(2,2,2,2,2)$ and $(a,b,c,d,1)_{a,b,c,d\geq 1}$.
The constraints on pseudoline number can be easily added to the program, and in both cases it found a proof-subtree.
The one found for the zonotope $(2,2,2,2,2)$ has \numprint{2034287} nodes, with \numprint{562480} leaves, and the one found for zonotopes $(a,b,c,d,1)_{a,b,c,d\geq 1}$ has \numprint{7045} nodes, with \numprint{1910} leaves.
We (hardly) achieved a double-checking for the zonotope $(2,2,2,2,2)$, by performing an exhaustive search among the \numprint{283323966772900} pairs of the tiling space.
Such an exhaustive search is of course impossible for the infinite family of zonotopes $(a,b,c,d,1)_{a,b,c,d\geq 1}$.\\

\noindent The following lemma summarizes all the above computer-obtained results:

\begin{lemma}
The rhombus tiling space of a five-bundle zonotope has no deficient tiling pair if and only if there is a size $1$ bundle, or only size $2$ bundles.
\end{lemma}

To conclude this subsection, let us mention that we also tested our program on the previous four-bundle case.
It confirms Lemma~\ref{lem:four_bundles}, with the proof-subtree found by our program having $72$ nodes, with $10$ leaves (that is, a bit more than the proof by hand).

\subsection{More bundles}

The case of rhombus tiling spaces with six or more bundles is quickly solved, since flip- and Hamming-distances turn out to be different for the simplest domain, namely the zonotope $(1,1,1,1,1,1)$, for which we found (by an exhaustive checking) $16$ deficient pairs among the \numprint{824464} pairs of the tiling space, each at flip-distance $12$ but at Hamming-distance $10$.
There are two pairs up to isometry, depicted on Fig.~\ref{fig:contrex_6}.
One of these deficient pairs  already appeared in \cite{felsner} in a similar context (Fig.~\ref{fig:contrex_6}, left, although only {\em singleton} bundles, that is, classical pseudoline arrangements, were there considered).
These deficient pairs yield our last lemma:

\begin{lemma}
Any rhombus tiling space of a six bundle zonotope admits at least one deficient tiling pair.
\end{lemma}

\begin{figure}[hbtp]
\centering
\includegraphics[width=0.95\textwidth]{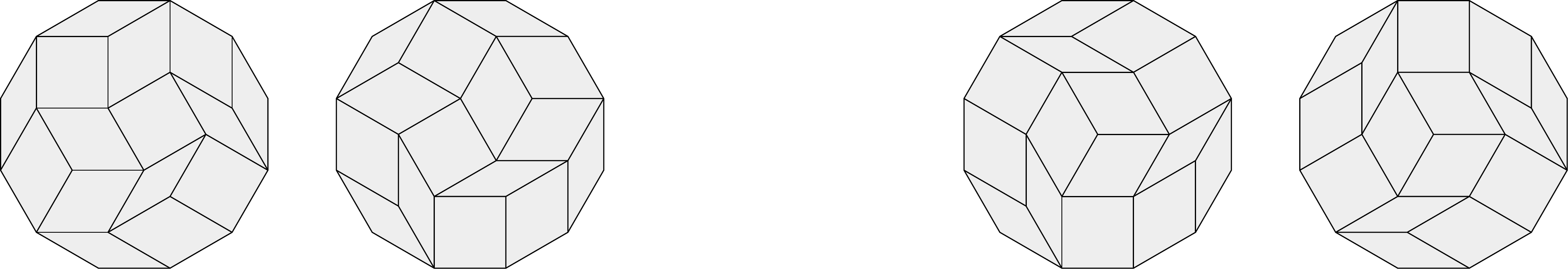}
\caption{The two deficient tiling pairs (up to symmetry) of $(1,1,1,1,1,1)$.}
\label{fig:contrex_6}
\end{figure}

\bibliographystyle{model1-num-names}

\end{document}